\documentclass[12pt,a4paper,headings=small,oneside]{scrartcl}

\usepackage{amsmath, amssymb, amsfonts}
\usepackage{amsthm}
\usepackage{mathtools}
\usepackage{graphics}
\usepackage{epstopdf}
\usepackage{natbib}
\usepackage{bbm}
\usepackage{bm}
\usepackage{setspace}
\usepackage{color}
\usepackage[hyphens]{url}
\usepackage{hyperref}
\usepackage{enumitem}
\usepackage{comment}

\newtheorem{theorem}{Theorem}

\newtheorem{proposition}[theorem]{Proposition}

\theoremstyle{definition}

\newtheorem{example}[theorem]{Example}
\newtheorem{definition}[theorem]{Definition}
\newtheorem{remark}[theorem]{Remark}

\newcommand{\I}{{\textbf 1}}

\newcommand{\C}{{\mathbb C}}

\DeclareMathOperator*{\argmin}{arg\,min}
\DeclareMathOperator{\crd}{crd}

\newcommand{\Ex}{\mathbb{E}} 

\usepackage[colorinlistoftodos,textsize=tiny]{todonotes}
\newcommand{\Comments}{1}
\newcommand{\mynote}[2]{\ifnum\Comments=1\textcolor{#1}{#2}\fi}
\newcommand{\mytodo}[2]{\ifnum\Comments=1%
	\todo[linecolor=#1!80!black,backgroundcolor=#1,bordercolor=#1!80!black]{#2}\fi}

\ifnum\Comments=1               
\fi

\begin{document}

\title{Circular Expectiles}
\author{Bernhard Klar\footnote{ bernhard.klar@kit.edu} \\
\small{Institute of Stochastics,} \\
\small{Karlsruhe Institute of Technology (KIT), Germany.}
}

\date{\today }
\maketitle

\begin{abstract}
In this work, we introduce circular expectiles as minimizers of an asymmetric circular loss function based on chord distance. In contrast to the linear expectile criterion, the resulting circular optimization problem is non-convex, so existence and uniqueness require a separate analysis. The construction extends linear expectiles to directional data while preserving the circular mean as the symmetric case corresponding to $\alpha=1/2$. We derive basic representations of the objective function and the associated identification function, and give a geometric interpretation that generalizes the corresponding representation for the circular mean. Furthermore, we prove the existence and uniqueness of the minimizers for distributions with positive density on the circle.
The empirical circular expectile is defined by using the sample circular mean as reference direction for the induced linear order on the circle. We prove the uniqueness of the empirical expectile, as well as its consistency and finite-dimensional asymptotic normality. Finally, we indicate possible applications to circular measures of dispersion, skewness, and symmetry diagnostics.
\end{abstract}


\section{Introduction} \label{sec1}

Directional and circular data analysis has a long history, and comprehensive background on measures of location, spread and shape on the circle can be found in the monographs by \cite{Fi:1993}, \cite{JS:2001}, and \cite{MJ:2000}.
These references cover the classical circular mean, median-type notions, measures
of concentration and dispersion, as well as descriptive quantities such as
circular skewness and kurtosis. 

Among measures of circular location, the most prominent notion is the
\emph{circular mean}, also called the extrinsic mean, which is characterized by
the direction of the mean resultant vector. Its large-sample theory is standard;
see, for example, \citet[p.76]{JS:2001} and \citet[p.76]{MJ:2000}, as well as \cite{Pe:2004}
for joint asymptotics of several key circular statistics, including skewness and kurtosis. Beyond the extrinsic mean, intrinsic or Fr\'echet-type means on the circle have been studied in detail by \cite{BP:2003, BP:2005}, by \cite{MQC:2012}, by \cite{Ho:2013}, and by \cite{HH:2015}, with particular emphasis on existence, uniqueness and asymptotic
properties. 

A second line of work concerns median-type location functionals based
on arc distance. On the circle, this includes the Mardia-Fisher median and
related spherical medians; relevant references are \cite{Fi:1985}, \citet[p.35]{Fi:1993}, \cite{LS:1992}, and \cite{Pu:1995}. These papers also connect circular location with notions of depth and ordering for directional data. In particular, Liu and Singh introduced
circular and spherical depth concepts, thereby providing a framework in which location, centrality and ordering can be treated simultaneously. 

More recently, quantile-type methodology for directional data has been developed in several directions. \cite{LSV:2014} proposed a notion of quantiles for directional data together with an angular Mahalanobis depth, thereby extending spatial-quantile ideas to the circular setting. On the regression side, \cite{DPT:2016} developed nonparametric circular quantile regression. Thus, in the literature on circular statistics there is by now a substantial body of work on location functionals based on means, medians, depth and quantiles.

In contrast, \emph{expectiles} have so far been studied almost exclusively in the linear setting. Introduced by \cite{NP:1987} as minimizers of an asymmetric quadratic loss function, they form a one-parameter family of non-central location functionals that interpolates through the mean and has attractive analytical and statistical properties. In particular, \cite{BKMG:2014} studied expectiles as generalized quantiles and coherent risk measures, \cite{HK:2016} developed their asymptotic theory in detail, and \cite{BKM:2018} related expectiles to omega ratios and stochastic ordering. 

This paper extends the idea to the circular setting. Starting from the circular mean as the symmetric benchmark, we introduce circular expectiles as minimizers of an asymmetric circular loss function based on chord distance. As in the linear case, expectiles are therefore measures of location related to  circular means and medians, as well as to circular quantiles and depths. However, unlike linear expectiles, the circular criterion is not convex in the location parameter. Even the symmetric loss function $1-\cos(x-\nu)$ has second derivative $\cos(x-\nu)$, which changes sign on the circle. Consequently, the standard arguments in convex M-estimation cannot be used to derive existence, uniqueness and asymptotic properties directly, and a direct analysis of the circular identification function is required instead.

The paper is organized as follows. The remainder of the introduction recalls the circular mean and the basic definition of linear expectiles. Section~2 introduces circular expectiles, derives basic representations and identification functions, and establishes existence and uniqueness results under suitable regularity assumptions. Section~3 studies the empirical version, including existence and uniqueness. Section~4 derives consistency and finite-dimensional asymptotic normality of empirical circular expectiles and concludes with a discussion of possible applications to circular measures of dispersion, skewness, and symmetry diagnostics.

\subsection{Circular mean}

Throughout the paper, the circular mean plays a fundamental role. We therefore recall its definition and basic properties in this subsection. 

Consider a circular random variable $X$, i.e. a random angle with values in $[-\pi,\pi)$. 
We identify every angle $\phi\in [-\pi,\pi)$ with the point $z=\zeta(\phi)=(\cos\phi,\sin\phi)$ on the unit circle $S^1=\{z\in \C:|z|=1\}$.
Let $\crd(\phi)$ denote the chord of $\phi$. The population circular mean set of $X$ or $Z=e^{iX}$ is defined as
\begin{align} \label{circ-mean}
	\mathcal M &= \argmin_{\nu \in [-\pi,\pi)} \Ex \left( \frac{\crd^2(X-\nu)}{2} \right) 
	= \argmin_{\nu \in [-\pi,\pi)} \Ex \left[1-\cos\left(X-\nu\right)\right].
\end{align}%
It is well known that $\mathcal M$ is a singleton if and only if $\Ex e^{iX}\neq 0$. In this case we denote its unique element by $\mu=\operatorname{Arg}_{[-\pi,\pi)}(\Ex Z)$. The direction $\mu$ satisfies $E\sin(X-\mu)=0$; the antipodal direction also satisfies this first-order equation, but corresponds to the maximum of the criterion.

Consider independent random variables $\theta_1,\ldots,\theta_n\sim X$, and write $Z_i=e^{i\theta_i}, i=1,\ldots,n$. The circular sample mean set is 
\begin{align} \label{emp-circ-mean}
  \widehat{\mathcal M}_n 
  &= \argmin_{\nu \in [-\pi,\pi)} \sum_{i=1}^n \left( 1-\cos\left(\theta_i-\nu\right)\right).
\end{align}
If $\bar{Z}_n=\frac1n\sum_{i=1}^n Z_i \neq 0$, then this set is a singleton and its element $\hat\mu_n$ is determined by
\[
 e^{i\hat\mu_n}=\frac{\bar Z_n}{|\bar Z_n|}, \quad
\hat\mu_n=\operatorname{Arg}_{[-\pi,\pi)}(\bar Z_n).
\]
For completeness, and since we will use the resulting representation below, we briefly recall a standard argument based on (\ref{circ-mean}). Define
\begin{align*} 
g(\nu) = 1 - \Ex[\cos(X-\nu)], \quad C=\Ex[\cos(X)], \quad S=\Ex[\sin(X)]
\end{align*}
Assuming $R=|\Ex Z|=\sqrt{C^2+S^2}>0$, and using $\cos(x-y)=\cos x\cos y+\sin x\sin y$ twice, we obtain
\begin{align} 
g(\nu) &= 1-C\cos(\nu)-S\sin(\nu)   \label{eq:g1} \\ 
&= 1 - R\cos(\nu-\phi), \nonumber
\end{align}
where $\phi=\operatorname{Arg}_{[-\pi,\pi)}(C+iS)$ . Thus the unique minimum is attained at $\nu=\phi$, whereas the unique maximum is attained at the antipodal direction.

A geometrical argument directly considers (\ref{eq:g1}). Minimizing $g$ is equivalent to maximizing the dot product $C\cos(\nu)+S\sin(\nu)$ between $(C,S)$ and $(\cos(\nu),\sin(\nu))$.
The maximum is attained if both vectors point in the same direction, which leads to the solution
\begin{align} \label{mu:geo}
    \left( \cos\phi, \sin\phi \right) &= \frac{(C,S)}{R}, 
\end{align}
the projection of $(C,S)$ onto the unit circle.

\subsection{Linear expectiles}

In this subsection, let $X$ be a real-valued random variable having a finite mean (denoted as $X \in L^1$). 
For $X\in L^2$, expectiles were introduced by \citet{NP:1987} as minimizers of an asymmetric quadratic loss:
\begin{equation}
e_X(\alpha)=\arg \min_{t\in \mathbb{R}}\left\{ E \ell_\alpha(X-t) \right\},  \label{exp_def_1}
\end{equation}%
where
\begin{equation*}
\ell_\alpha(x) = 
\begin{cases}
\alpha x^2, & \mbox{ if } x \ge 0, \\
(1-\alpha) x^2, & \mbox{ if } x < 0,%
\end{cases}%
\end{equation*}
and $\alpha \in (0,1)$. 
For $X\in L^{1}$, the objective in \eqref{exp_def_1} is replaced by the centered version
\begin{equation}
e_X(\alpha) = \arg \min_{t\in \mathbb{R}} \left\{ E\left[ \ell_\alpha(X-t) -
\ell_\alpha(X) \right]\right\}.  \label{exp_def_3}
\end{equation}
The minimizer in \eqref{exp_def_1} or \eqref{exp_def_3} is always unique and is identified by the first order condition%
\begin{equation*}
\alpha E\left( X-e_{X}(\alpha )\right) _{+}=(1-\alpha )E\left(
X-e_{X}(\alpha )\right) _{-},  
\end{equation*}
where $x_{+}=\max \{x,0\}$, $x_{-}=\max \{-x,0\}$.
This is equivalent to characterizing expectiles via an identification function, which, for any $\alpha \in (0, 1)$ is defined by
\begin{equation*}
	I_\alpha(x, y) = \alpha (y - x) \mathbbm{1}_{\{y \geq x\}} - (1 - \alpha) (x - y) \mathbbm{1}_{\{y < x\}}
\end{equation*}
for $x, y \in \mathbb{R}$. The $\alpha$-expectile of a random variable $X \in L^1$ is then the unique solution of
\begin{equation*}
	E I_\alpha(t, X) = 0, \quad t \in \mathbb{R}.
\end{equation*}
For $\alpha=1/2$, this gives $e_X(1/2)=\Ex X$. 
Further properties of expectiles, in particular their monotonicity, continuity and interpretation as generalized quantiles and risk measures, can be found in \citet{BKMG:2014}, \citet{BKM:2018}, and \citet{Ziegel:2016}. For asymptotic results on empirical expectiles, we refer to \citet{NP:1987}, \citet{HK:2016}, and \citet{KZ:2017}; see also \citet{SK:2012} for related work on expectile regression.

\section{Circular expectiles}

In this section, we introduce circular expectiles as minimizers of an asymmetric circular loss, and develop their basic representations and structural properties.

\begin{definition}
Let $X$ be a circular random variable with values in $[-\pi,\pi)$. Assume $R=|\Ex e^{iX}|>0$ and let $\mu=\arg(Ee^{iX})$ denote the (unique)  circular mean direction. Assume further that
\begin{align} \label{antipode}
P(e^{iX}=-e^{i\mu})=0.
\end{align}
Let $X_{\mu}$ denote the unique representative of $X$ modulo $2\pi$ belonging to $(\mu-\pi,\mu+\pi)$.
For $\alpha\in(0,1)$, the circular $\alpha$-expectile of $X$ is defined as
\begin{align*} 
\mathcal{E}_{\alpha} = \argmin_{\nu \in (\mu-\pi,\mu+\pi)} \, g_{\alpha}(\nu),
\end{align*}
where
\begin{align*}
g_{\alpha}(\nu) &= \Ex \big[ \alpha (1-\cos(X_{\mu}-\nu)) \, \I_{\{ X_{\mu} \in (\nu,\mu+\pi) \}} 
+ (1-\alpha) (1-\cos(X_{\mu}-\nu)) \, \I_{\{ X_{\mu} \in (\mu-\pi,\nu] \}} \big].
\end{align*}
If $\mathcal{E}_{\alpha}$ is a singleton, we denote its element  by $\mu_{\alpha}$.
\end{definition}

\begin{remark}
\begin{enumerate}[label=(\roman*)]
\item 
The assumption in \eqref{antipode} avoids ambiguity in the definition at the antipodal point. It is automatically satisfied for absolutely continuous distributions, which are the main focus of this paper.
\item 
We have $\mu_{1/2} = \mu$ as in the linear case. The motivation behind this definition is also analogous: since $1-\cos(\theta-\nu)$ measures the distance between angles $\theta$ and $\nu$, the quantity $1-\Ex[\cos(X-\nu)]$ is a measure of the variability of $X$ about the direction $\nu$. In contrast to the (circular) variance, these deviations are priced asymmetrically, depending on the sign of the deviation.
\end{enumerate}
\end{remark}

\subsection{Definition and basic representations}

In the remainder of this section, all random angles are interpreted via their unique representatives in the interval $(\mu-\pi,\mu+\pi)$. For simplicity, we write $X$ instead of $X_\mu$.
Accordingly, all inequalities are understood with respect to the linear order on $(\mu-\pi,\mu+\pi)$.
Define
\[
F(\nu)=\mathbb{E}[\mathbf{1}_{\{X\le \nu\}}], \qquad
C=\mathbb{E}[\cos X], \qquad
S=\mathbb{E}[\sin X],
\]
and
\[
C_\nu^- = \mathbb{E}[\cos(X)\mathbf{1}_{\{X\le \nu\}}], \qquad
S_\nu^- = \mathbb{E}[\sin(X)\mathbf{1}_{\{X\le \nu\}}],
\]
together with
\[
C_\nu^+ = C - C_\nu^-, \qquad  S_\nu^+ = S - S_\nu^-.
\]
Then the objective function can be written as
\begin{align*} 
g_\alpha(\nu) &= \alpha + (1-2\alpha)F(\nu) - C_\alpha(\nu)\cos\nu - S_\alpha(\nu)\sin\nu,
\end{align*}
where
\[
C_\alpha(\nu)=\alpha C + (1-2\alpha)C_\nu^-,
\quad
S_\alpha(\nu)=\alpha S + (1-2\alpha)S_\nu^-.
\]
The boundary values are
\[
\lim_{\nu\downarrow \mu-\pi} g_\alpha(\nu)=\alpha(1+R),
\quad
\lim_{\nu\uparrow \mu+\pi} g_\alpha(\nu)=(1-\alpha)(1+R),
\]
where $R=\Ex[\cos(X-\mu)]=C\cos\mu+S\sin\mu=\sqrt{C^2+S^2}$ denotes the mean resultant length.

\subsection{Identification function}

Assume that $X$ has a density $f$. Define the function
\[
I_\alpha(\nu,x) = \alpha \sin(\nu-x)\mathbf{1}_{\{x>\nu\}}
+ (1-\alpha)\sin(\nu-x)\mathbf{1}_{\{x\le \nu\}},
\]
and
\[
I_\alpha(\nu,F)=\mathbb{E}[I_\alpha(\nu,X)].
\]
We now show that $I_\alpha(\nu,F)$ is the derivative of $g_\alpha(\nu)$.
Differentiating under the integral sign yields
\[
g_\alpha'(\nu)
= \alpha \mathbb{E}\big[\sin(\nu-X)\mathbf{1}_{\{X>\nu\}}\big]
+ (1-\alpha)\mathbb{E}\big[\sin(\nu-X)\mathbf{1}_{\{X\le\nu\}}\big],
\]
since $\frac{d}{d\nu}(1-\cos(x-\nu))=\sin(\nu-x)$ and the contribution from the
moving boundary $x=\nu$ vanishes because $1-\cos(0)=0$. Hence,
\[
g_\alpha'(\nu)=I_\alpha(\nu,F).
\]
Thus, $I_\alpha(\nu,F)$ serves as an identification function for the circular $\alpha$-expectile. In particular, any minimizer satisfies the first-order condition
\[
I_\alpha(\nu,F)=0.
\]
We next derive alternative representations of $I_\alpha(\nu,F)$. Using integration by parts, we obtain
\begin{align*}
I_{\alpha}(\nu,F) = -\alpha \int_{\nu}^{\mu+\pi} (1-F(x)) \cos(x-\nu)\, dx 
+ (1-\alpha) \int_{\mu-\pi}^{\nu} F(x) \cos(x-\nu)\, dx.
\end{align*}
A further representation is obtained by expressing the integrals in terms of the
trigonometric partial moments:
\begin{align}
I_{\alpha}(\nu,F) &= \alpha \bigl( C_{\nu}^+\sin\nu - S_{\nu}^+\cos\nu \bigr)
+ (1-\alpha)\bigl( C_{\nu}^-\sin\nu - S_{\nu}^-\cos\nu \bigr) \nonumber \\
&= C_{\alpha}(\nu)\sin\nu - S_{\alpha}(\nu)\cos\nu. \label{ident-fct}
\end{align}
The boundary values are
\[
\lim_{\nu\downarrow \mu-\pi} I_{\alpha}(\nu,F)=0, \qquad
\lim_{\nu\uparrow \mu+\pi} I_{\alpha}(\nu,F)=0.
\]
Indeed, 
\begin{align*}
C_\alpha(\nu)\to \alpha C, \quad S_\alpha(\nu)\to \alpha S  
&\quad \text{as } \nu\downarrow \mu-\pi, \\
C_\alpha(\nu)\to (1-\alpha)C, \quad S_\alpha(\nu)\to (1-\alpha)S  
&\quad \text{as } \nu\uparrow \mu+\pi, 
\end{align*}
and
\[
C\sin(\mu\pm\pi)-S\cos(\mu\pm\pi) = -\,\bigl(C\sin\mu-S\cos\mu\bigr)=0.
\]

\subsection{Geometric interpretation}

By \eqref{ident-fct}, the equation $I_\alpha(\nu,F)=0$ is equivalent to
\[
(C_\alpha(\nu),S_\alpha(\nu)) \cdot (\sin\nu,-\cos\nu)=0,
\]
i.e., $(C_\alpha(\nu),S_\alpha(\nu))$ is orthogonal to $(\sin\nu,-\cos\nu)$.
Since $(\sin\nu,-\cos\nu)$ is obtained by a $90^\circ$ rotation of $(\cos\nu,\sin\nu)$, it follows that
$(C_\alpha(\nu),S_\alpha(\nu))$ is collinear with $(\cos\nu,\sin\nu)$.
At a stationary point with
\[
(C_\alpha(\nu),S_\alpha(\nu)) = \lambda (\cos\nu,\sin\nu),
\]
one has (compare \eqref{ident-derivative} below)
\begin{align*}  
g''_\alpha(\nu) = I'_\alpha(\nu,F) = C_\alpha(\nu)\cos(\nu)+S_\alpha(\nu)\sin(\nu) = \lambda.
\end{align*}
Since $g''_\alpha(\nu)=\lambda$, a strict local minimizer therefore requires $\lambda>0$, so the vectors point in the same direction.
Hence $\mu_\alpha$ satisfies
\[
\frac{(C_\alpha(\mu_\alpha),S_\alpha(\mu_\alpha))}{\sqrt{C_\alpha(\mu_\alpha)^2+S_\alpha(\mu_\alpha)^2}}
= (\cos\mu_\alpha,\sin\mu_\alpha),
\]
provided the denominator is nonzero. This representation is the counterpart of \eqref{mu:geo} for $\alpha\neq 1/2$.

\subsection{Uniqueness of the theoretical circular expectile}

The following theorem establishes existence, uniqueness, and basic structural properties of circular expectiles. Since circular expectiles are equivariant under rotations, we may assume without loss of generality that $\mu=0$. In this case, $S=0$ and $R=C\in [0,1].$

Throughout, we assume $0<\alpha<1/2$. The case $1/2<\alpha<1$ is treated in Remark~\ref{rem:upper-expectiles}. 

\begin{theorem} \label{thm:main-lower-expectile}
Let $X$ be a circular random variable with values in $[-\pi,\pi)$ and circular mean
$\mu=0$ (so that $S=\mathbb{E}[\sin X]=0$), and let $C=\mathbb{E}[\cos X]>0$.
Assume that $X$ has a density $f$ which is positive a.e.\ on $(-\pi,\pi)$.
For $\alpha\in(0,1/2)$, define
\[
C_\alpha(\nu)=\alpha C+(1-2\alpha)C_\nu^-,
\qquad
S_\alpha(\nu)=(1-2\alpha)S_\nu^-,
\]
where
\[
C_\nu^- = \Ex[\cos(X)\mathbf 1_{\{X\le \nu\}}],
\qquad
S_\nu^- = \Ex[\sin(X)\mathbf 1_{\{X\le \nu\}}],
\]
and let
\[
I_\alpha(\nu,F)=C_\alpha(\nu)\sin\nu-S_\alpha(\nu)\cos\nu,
\qquad \nu\in(-\pi,\pi).
\]
Then the following hold.

\begin{enumerate}[label=(\alph*)]
\item 
For every $\alpha\in(0,1/2)$, the equation $I_\alpha(\nu,F)=0$ has a unique solution $q(\alpha)\in(-\pi,0)$. Moreover, $q(\alpha)$ is the unique minimizer of $g_{\alpha}$ or, in terms of the expectile notation, $\mathcal{E}_{\alpha}=\{q(\alpha)\}$.
\item 
The function $q:(0,1/2)\to(-\pi,0)$ is strictly increasing.
\item 
Define
\[
\alpha^*:=\inf\Bigl\{\alpha\in[0,1/2]:\, C_\alpha(-\pi/2)\ge 0\Bigr\}.
\]
Then $\alpha^*\in(0,1/2)$ is given by
\[
\alpha^*=\frac{-\,C^-_{-\pi/2}}{\,C-2C^-_{-\pi/2}\,}.
\]
Moreover, $q(\alpha)\in(-\pi,-\pi/2)$ for $\alpha<\alpha^*$, $q(\alpha^*)=-\pi/2$,
and $q(\alpha)\in(-\pi/2,0)$ for $\alpha>\alpha^*$.
\item 
One has $\lim_{\alpha\uparrow 1/2} q(\alpha)=0$ and 
$\lim_{\alpha\downarrow 0} q(\alpha)=-\pi$.
\end{enumerate}
\end{theorem}

\begin{proof}
\noindent \emph{Preliminary identities.}
We have 
\[
C_\alpha'(\nu)=(1-2\alpha)\cos(\nu)f(\nu), \qquad
S_\alpha'(\nu)=(1-2\alpha)\sin(\nu)f(\nu),
\]
and
\begin{align}  
I'_\alpha(\nu,F)=C_\alpha(\nu)\cos(\nu)+S_\alpha(\nu)\sin(\nu). \label{ident-derivative}
\end{align}
The identities for $C_\alpha'$ and $S_\alpha'$ hold a.e.; the displayed expression for $I'_\alpha$ follows everywhere by differentiating the representation of $I_\alpha$ in  \eqref{ident-fct}, since the moving-boundary term vanishes.

Since $(S_\nu^-)'=\sin(\nu)f(\nu)$ and $f>0$ a.e., the function $S_\nu^-$ is strictly decreasing on $(-\pi,0)$, strictly increasing on $(0,\pi)$, and satisfies $S^-_{-\pi}=S^-_\pi=0$. Hence
\begin{align}
S_\alpha(\nu)<0 \qquad \text{for all }\nu\in(-\pi,\pi). \label{unique-1}
\end{align}
Moreover, $C_\alpha$ is strictly decreasing on $(-\pi,-\pi/2)$, strictly increasing on $(-\pi/2,\pi/2)$, and strictly decreasing on $(\pi/2,\pi)$.

\begin{enumerate}[label=(\alph*)]
\item \emph{Existence and uniqueness.}

On the interval $(-\pi/2,\pi/2)$, define
\[
J_\alpha(\nu):=\frac{I_\alpha(\nu,F)}{\cos(\nu)}
= C_\alpha(\nu)\tan(\nu)-S_\alpha(\nu).
\]
Since $\cos(\nu)>0$ there, $I_\alpha$ and $J_\alpha$ have the same zeros. A direct
calculation yields
\begin{align*}
J'_\alpha(\nu)=\frac{C_\alpha(\nu)}{\cos^2(\nu)}.
\end{align*}

\medskip\noindent
\emph{Case 1: $C_\alpha(-\pi/2)>0$.}

Then $C_\alpha(\nu)>0$ for all $\nu\in(-\pi,\pi)$, and hence $J_\alpha$ is strictly increasing on $(-\pi/2,\pi/2)$. Since
\[
J_\alpha(\nu)\to -\infty \quad (\nu\downarrow -\pi/2),
\qquad
J_\alpha(0)=-S_\alpha(0)>0
\]
by \eqref{unique-1}, there exists exactly one zero in $(-\pi/2,0)$.

On $(-\pi,-\pi/2)$, one has $\sin(\nu)<0$, $\cos(\nu)<0$, $C_\alpha(\nu)>0$, and
$S_\alpha(\nu)<0$, hence
\[
I_\alpha(\nu,F)=C_\alpha(\nu)\sin(\nu)-S_\alpha(\nu)\cos(\nu)<0.
\]
On $(0,\pi/2)$, we have $\sin(\nu)>0$, $\cos(\nu)>0$, $C_\alpha(\nu)>0$, and
$S_\alpha(\nu)<0$, hence $I_\alpha(\nu,F)>0$. On $(\pi/2,\pi)$,
\[
I'_\alpha(\nu,F)<0,
\qquad
\lim_{\nu\uparrow\pi} I_\alpha(\nu,F)=0,
\]
so $I_\alpha(\nu,F)>0$ there as well. Hence the zero is unique.

\medskip\noindent
\emph{Case 2: $C_\alpha(-\pi/2)=0$.}

Then $I_\alpha(-\pi/2,F)=0$. Since $C_\alpha$ is strictly increasing on $(-\pi/2,\pi/2)$, we have $C_\alpha(\nu)>0$ for all $\nu\in(-\pi/2,\pi/2)$. On $(-\pi/2,0)$,
\[
I_\alpha'(\nu,F)=C_\alpha(\nu)\cos\nu+S_\alpha(\nu)\sin\nu>0,
\]
because $C_\alpha(\nu)>0$, $\cos\nu>0$, $S_\alpha(\nu)<0$, and $\sin\nu<0$. Hence $I_\alpha(\nu,F)>0$ on $(-\pi/2,0)$. On $(0,\pi/2)$, the representation
\[
I_\alpha(\nu,F)=C_\alpha(\nu)\sin\nu-S_\alpha(\nu)\cos\nu
\]
shows directly that $I_\alpha(\nu,F)>0$. Also $I_\alpha(0,F)=-S_\alpha(0)>0$. Therefore there are no further zeros in $(-\pi/2,\pi/2)$. Exactly the same arguments as in Case~1 show that $I_\alpha(\nu,F)<0$ on $(-\pi,-\pi/2)$ and $I_\alpha(\nu,F)>0$ on $[\pi/2,\pi)$.

\medskip\noindent
\emph{Case 3: $C_\alpha(-\pi/2)<0$.}

Then $C_\alpha$ has a unique zero $\tau_\alpha\in(-\pi,-\pi/2)$ and a unique zero $\sigma_\alpha\in(-\pi/2,\pi/2)$. 

On $(-\pi,\tau_\alpha]$, one has $C_\alpha(\nu)\ge 0$, $\sin(\nu)<0$, $\cos(\nu)<0$,
and $S_\alpha(\nu)<0$, hence $I_\alpha(\nu,F)<0$.

On $(\tau_\alpha,-\pi/2)$, we have $C_\alpha(\nu)<0$, $S_\alpha(\nu)<0$, $\cos(\nu)<0$, and $\sin(\nu)<0$, so $I'_\alpha(\nu,F)>0$.
Moreover,
\[
I_\alpha(\tau_\alpha,F)<0, \qquad I_\alpha(-\pi/2,F)>0,
\]
so there is exactly one zero in $(\tau_\alpha,-\pi/2)$.

On $(-\pi/2,\pi/2)$, $J_\alpha$ decreases on $(-\pi/2,\sigma_\alpha)$ and increases on $(\sigma_\alpha,\pi/2)$. Since
\[
J_\alpha(\sigma_\alpha)=-S_\alpha(\sigma_\alpha)>0,
\]
it follows that $J_\alpha(\nu)>0$ for all $\nu\in(-\pi/2,\pi/2)$, hence no further zero exists. Again, $I_\alpha(\nu,F)>0$ on $(\pi/2,\pi)$. This proves uniqueness.

In all three cases, $I_\alpha(\nu,F)$ changes sign exactly once, from negative to positive, at its unique zero $q(\alpha)\in(-\pi,0)$. Since $g_\alpha'(\nu)=I_\alpha(\nu,F)$, it follows that $q(\alpha)$ is the unique minimizer of $g_\alpha$.

\item \emph{Monotonicity.}

By (a), $q(\alpha)$ is well-defined by $I_\alpha(q(\alpha),F)=0$. At every root,
$I'_\alpha(q(\alpha),F)>0$, including the boundary case $q(\alpha)=-\pi/2$, since
\[
I'_\alpha(-\pi/2,F)=-\,S_\alpha(-\pi/2)>0
\]
by \eqref{unique-1}. Hence the implicit function theorem applies, so $q$ is continuously
differentiable. Moreover, Proposition~\ref{prop:derivative} yields $q'(\alpha)>0$ for all
$\alpha\in(0,1/2)$.

\item \emph{Threshold.}
The derivative of $C_\alpha(-\pi/2)=\alpha C+(1-2\alpha)C^-_{-\pi/2}$ is 
\[
\frac{d}{d\alpha}C_\alpha(-\pi/2)=C-2C^-_{-\pi/2}.
\]
Since $C>0$ and $C^-_{-\pi/2}<0$, $C_\alpha(-\pi/2)$ is strictly increasing in $\alpha$. Moreover,
\[
C_0(-\pi/2)=C^-_{-\pi/2}<0, \quad C_{1/2}(-\pi/2)=C/2>0.
\]
Therefore, there exists a unique $\alpha^*\in(0,1/2)$ such that $C_{\alpha^*}(-\pi/2)=0$, namely
\[
\alpha^*=\frac{-\,C^-_{-\pi/2}}{C-2C^-_{-\pi/2}}.
\]
The remaining assertions follow directly from part~(a): if $\alpha<\alpha^*$, then
$C_\alpha(-\pi/2)<0$, so $q(\alpha)\in(-\pi,-\pi/2)$; if $\alpha=\alpha^*$, then
$q(\alpha)=-\pi/2$; and if $\alpha>\alpha^*$, then $C_\alpha(-\pi/2)>0$, so
$q(\alpha)\in(-\pi/2,0)$.

\item \emph{Limit as $\alpha\uparrow 1/2$.}

Since $q(\alpha)$ is increasing and bounded above by $0$, the limit exists. By continuity,
$I_{1/2}(\nu,F)=C\sin(\nu)/2$, so the only zero in $(-\pi,0]$ is $\nu=0$.

\medskip
\emph{Limit as $\alpha\downarrow 0$.}

Since $q(\alpha)$ is increasing in $\alpha$ and $q(\alpha)\in(-\pi,0)$ for all
$\alpha\in(0,1/2)$, the limit
\[
q_0:=\lim_{\alpha\downarrow 0} q(\alpha)
\]
exists in $[-\pi,0)$. We claim that $q_0=-\pi$.

Suppose, to the contrary, that $q_0>-\pi$. By continuity of $I_\alpha(\nu,F)$ in
$(\alpha,\nu)$, we have $I_0(q_0,F)=0$. But
\[
I_0(q_0,F)=\int_{-\pi}^{q_0}\sin(q_0-x)f(x)\,dx > 0,
\]
since $f>0$ a.e.\ on $(-\pi,q_0)$. This is a contradiction. Hence $q_0=-\pi$.

%
\end{enumerate}
\end{proof}

\begin{remark}[Upper circular expectiles] \label{rem:upper-expectiles}
For $\alpha\in(1/2,1)$ the result follows by reflection. Let $Y=-X$ and denote its distribution function by $F_Y$. Then $Y$ has circular mean $0$ and satisfies the assumptions of Theorem~\ref{thm:main-lower-expectile}. Moreover,
\[
I_\alpha(\nu,F_X)=-I_{1-\alpha}(-\nu,F_Y).
\]
Hence $\nu$ solves $I_\alpha(\nu,F_X)=0$ if and only if $-\nu$ solves
$I_{1-\alpha}(\cdot,F_Y)=0$. Applying Theorem~\ref{thm:main-lower-expectile} to $Y$ with
parameter $1-\alpha\in(0,1/2)$ shows that $I_\alpha(\nu,F_X)=0$ has a
unique solution $q(\alpha)\in(0,\pi)$, which is the unique minimizer of
$g_\alpha$.

More precisely, let $\alpha_Y^*$ be the threshold from Theorem~\ref{thm:main-lower-expectile} applied to $Y$, and set $\widetilde\alpha^*:=1-\alpha_Y^*$. Then
\[
q(\alpha)\in(0,\pi/2) \;\; \text{for } \alpha<\widetilde\alpha^*, \qquad
q(\widetilde\alpha^*)=\pi/2, \qquad  
q(\alpha)\in(\pi/2,\pi) \;\; \text{for } \alpha>\widetilde\alpha^*.
\]
Moreover,
\[
\lim_{\alpha\downarrow1/2}q(\alpha)=0,
\qquad
\lim_{\alpha\uparrow1}q(\alpha)=\pi .
\]
\end{remark}

The following proposition gives the derivative of the circular expectile.

\begin{proposition}
\label{prop:derivative}
Let $\alpha\in(0,1/2)$, and let $q(\alpha)\in(-\pi,0)$ denote the unique solution of
$I_\alpha(\nu,F)=0$. Then $q$ is continuously differentiable and
\begin{align} \label{expect-deriv}
q'(\alpha) = \frac{2D\left(q(\alpha)\right)-C\sin q(\alpha)}{I_\alpha'(q(\alpha),F)},
\end{align}
where $D(\nu)=C_\nu^-\sin\nu - S_\nu^-\cos\nu$, and $I'_\alpha(\nu,F)$ is given by \eqref{ident-derivative}. Equivalently, one has the representation
\begin{align} \label{expect-deriv2}
q'(\alpha) = \frac{-C \sin q(\alpha)}{(1-2\alpha)\,I_\alpha'(q(\alpha),F)}.
\end{align}
In particular, $q'(\alpha)>0$ for all $\alpha\in(0,1/2)$.
\end{proposition}

\begin{proof}
By Theorem~\ref{thm:main-lower-expectile}(a), the equation $I_\alpha(q(\alpha),F)=0$
defines $q(\alpha)$ uniquely for $\alpha\in(0,1/2)$. Moreover, by the proof of
Theorem~\ref{thm:main-lower-expectile}(a), one has $I'_\alpha(q(\alpha),F)>0$.
Hence the implicit function theorem implies that $q$ is continuously differentiable and
\begin{align} \label{qprime}
q'(\alpha) = -\frac{\partial_\alpha I_\alpha(q(\alpha),F)}{I'_\alpha(q(\alpha),F)}.
\end{align}
Writing $I_\alpha(\nu,F)=\alpha C\sin\nu+(1-2\alpha)D(\nu)$, 
where $D(\nu):=C_\nu^-\sin\nu-S_\nu^-\cos\nu$, we obtain
\[
\partial_\alpha I_\alpha(\nu,F)=C\sin\nu-2D(\nu).
\]
Evaluating at $\nu=q(\alpha)$ gives \eqref{expect-deriv}.
Using $I_\alpha(q(\alpha),F)=0$, we have
\[
D(q(\alpha))=-\frac{\alpha}{1-2\alpha}C\sin q(\alpha),
\]
and therefore
\[
\partial_\alpha I_\alpha(q(\alpha),F) = \frac{C\sin q(\alpha)}{1-2\alpha}<0.
\]
Substituting this into \eqref{qprime} yields \eqref{expect-deriv2}.
Since $I'_\alpha(q(\alpha),F)>0$, it follows that $q'(\alpha)>0$.
\end{proof}

\begin{remark}
The derivative formula for circular expectiles has the same implicit-function structure as in the linear case: it is obtained as the negative ratio of the derivative of the identification function with respect to the asymmetry parameter and its derivative with respect to the location parameter.

In the linear setting (see Proposition~1(iii) in \citet{HK:2016}), this yields a ratio involving integrated tail functionals in the numerator and a local cdf-based weight in the denominator. In the circular setting, these quantities are replaced by trigonometric moment terms and the local slope
$I_\alpha'(q(\alpha),F)$ at the root. 
\end{remark}

\begin{example}[Uniform distribution on an arc]
\label{ex:uniform-subarc}
Let $X\sim \mathrm{Unif}(-d,d)$ with $0<d<\pi$. Then $X$ is a circular random variable
whose support is a strict subset of $[-\pi,\pi)$, and its circular mean is $\mu=0$.
In this case, $C=\sin d/d$, and, for $\nu\in (-d,d)$,
\[
C_\nu^-= \frac{\sin \nu+\sin d}{2d},  \qquad
S_\nu^-= \frac{\cos d-\cos \nu}{2d}.
\]
Therefore, on the support $(-d,d)$,
\begin{align*}
I_\alpha(\nu,F) &= \alpha C \sin\nu + (1-2\alpha) \left(C_\nu^-\sin\nu-S_\nu^-\cos\nu \right) \\ 
 &= \frac{1}{2d}\left(\sin d\,\sin \nu + (1-2\alpha)\left(1-\cos d\,\cos \nu\right)\right).
\end{align*}
Introducing the substitution
\[
t=\tan\frac{\nu}{2}, \quad \sin\nu=\frac{2t}{1+t^2}, \quad \cos\nu=\frac{1-t^2}{1+t^2},
\]
and multiplying by $1+t^2$, the equation $I_\alpha(\nu,F)=0$ is equivalent to
\[
2t\sin d + (1-2\alpha) \left( (1-\cos d) + (1+\cos d)t^2 \right) = 0.
\]
Using $\sin d=2\sin(d/2)\cos(d/2),$ $1-\cos d=2\sin^2(d/2),$ $1+\cos d=2\cos^2(d/2)$, and dividing by $2\cos^2(d/2)$, we obtain the quadratic equation
\[
(1-2\alpha)t^2 + 2t\tan(d/2) + (1-2\alpha)\tan^2(d/2) = 0.
\]
Solving the quadratic equation gives
\[
t = \tan\frac{\nu}{2} 
= \frac{\sqrt{\alpha}-\sqrt{1-\alpha}}{\sqrt{\alpha}+\sqrt{1-\alpha}} \ \tan\frac d2,
\]
where we select the root corresponding to $\nu\in(-d,d)$. Hence
\[
\mu_\alpha = 2\arctan\left( 
    \frac{\sqrt{\alpha}-\sqrt{1-\alpha}}{\sqrt{\alpha}+\sqrt{1-\alpha}} \ \tan\frac d2 \right)
\]
for $\alpha\in(0,1)$. In particular, $\mu_\alpha\in(-d,d)$, $\lim_{\alpha\downarrow 0} \mu_\alpha=-d$, $\lim_{\alpha\uparrow 1} \mu_\alpha=d$, and $\mu_{1-\alpha}=-\mu_\alpha$ by symmetry.

There are no further zeros outside the support: for $\nu<-d$, $I_\alpha(\nu,F)=\alpha C\sin\nu<0,$
whereas for $\nu>d$, $I_\alpha(\nu,F)=(1-\alpha)C\sin\nu>0$.
Hence the zero displayed above is the unique minimizer.

Note that the density of $X$ is not strictly positive on $(-\pi,\pi)$, so the assumptions of Theorem~\ref{thm:main-lower-expectile} are not satisfied. Nevertheless, the circular expectile is uniquely defined, illustrating that the full-support assumption is not essential in general.

Moreover, for $d>\pi/2$, we obtain $C_{-\pi/2}^-=(\sin d-1)/(2d)$. Thus, the threshold parameter from
Theorem~\ref{thm:main-lower-expectile} is given by $\alpha^*=(1-\sin d)/2\in(0,1/2)$.
A direct calculation shows that
\[
0<\alpha<\alpha^* \;\Rightarrow\; \mu_\alpha\in(-d,-\pi/2),
\qquad
\alpha^*<\alpha<1/2 \;\Rightarrow\; \mu_\alpha\in(-\pi/2,0),
\]
so the qualitative behavior described in Theorem~\ref{thm:main-lower-expectile}
remains valid despite the lack of full support.
\end{example}


\section{Empirical circular expectiles}

For an iid sample $\theta=(\theta_1,\ldots,\theta_n), n\geq 2,$ from an absolutely continuous circular distribution on $[-\pi,\pi)$, let $\hat\mu_n$ denote the sample circular mean.
We represent the observations in the interval $(\hat\mu_n-\pi,\hat\mu_n+\pi)$ and, for
simplicity, keep the notation $\theta_1,\ldots,\theta_n$ for these representatives.
The empirical circular $\alpha$-expectile set is defined by
\begin{align*}
\widehat{\mathcal E}_{n,\alpha} = \argmin_{\nu \in (\hat\mu_n-\pi,\hat\mu_n+\pi)} g_{n,\alpha}(\nu,\theta),
\end{align*}
where 
\begin{align*} 
g_{n,\alpha}(\nu,\theta) = \frac1n \sum_{i=1}^n 
  \Big\{ & \alpha (1-\cos(\theta_i-\nu)) \, \I_{\{ \theta_i \in (\nu,\hat\mu_n+\pi) \}} \\  
  & + (1-\alpha) (1-\cos(\theta_i-\nu)) \, \I_{\{ \theta_i \in (\hat\mu_n-\pi,\nu) \}} \Big\}.
\end{align*}
If this set is a singleton, we denote its element by $\hat\mu_{n,\alpha}$.
Note that the definition of $\hat\mu_{n,\alpha}$ depends on the sample circular mean $\hat\mu_n$, which serves as the reference point for the linearization of the circle. This dependence is inherent in the construction, since the notion of ordering on the circle requires a choice of origin.

Such a reference point is standard in the directional setting. For instance, the definition of projection quantiles in \cite{LSV:2014} uses a consistent estimator of a spherical median, and similar constructions appear in other notions of circular quantiles.

\subsection{Empirical identification function}

Define the empirical identification function by
\[
I_{n,\alpha}(\nu,\theta) = \frac1n \sum_{i=1}^n \left(
\alpha \sin(\nu-\theta_i)\,\mathbf{1}_{\{\theta_i\in(\nu,\hat\mu_n+\pi)\}}
+ (1-\alpha)\sin(\nu-\theta_i)\,\mathbf{1}_{\{\theta_i\in(\hat\mu_n-\pi,\nu)\}} \right).
\]
Then the function $g_{n,\alpha}(\nu,\theta)$ is differentiable on $(\hat\mu_n-\pi,\hat\mu_n+\pi)$ and satisfies
\[
\frac{\partial}{\partial \nu} g_{n,\alpha}(\nu,\theta) = I_{n,\alpha}(\nu,\theta).
\]
Away from the sample points, the indicators remain constant as functions of $\nu$,
and differentiation of $1-\cos(\theta_i-\nu)$ yields $\sin(\nu-\theta_i)$.
At the sample points $\nu=\theta_i$, the indicators change discontinuously. 
However, this does not destroy differentiability, since the function
$1-\cos(\theta_i-\nu)$ vanishes quadratically at $\nu=\theta_i$, so that the
corresponding summand has derivative zero at this point from both sides.

Thus, $I_{n,\alpha}(\nu,\theta)$ serves as the empirical identification function. 
For $\alpha\neq 1/2$, the function $g_{n,\alpha}(\nu,\theta)$ is not
twice differentiable at the sample points, reflecting the piecewise structure induced
by the indicators. The proof of the following theorem is provided in the appendix.

\begin{theorem}[Uniqueness of the empirical circular expectile] \label{prop:empirical}
Assume that $\theta_1,\ldots,\theta_n$ are sampled from an absolutely continuous circular distribution. Then, with probability one, the sample circular mean $\hat\mu_n$ is uniquely defined and the ordered sample satisfies
\[
-\pi =: \theta_{(0)} < \theta_{(1)} < \cdots < \theta_{(n)} < \theta_{(n+1)} := \pi.
\]
Without loss of generality, rotate the sample so that $\hat\mu_n=0$. In particular,
\[
C_n:=\frac1n\sum_{i=1}^n \cos(\theta_i)>0, \qquad
S_n:=\frac1n\sum_{i=1}^n \sin(\theta_i)=0.
\]

For $k=0,\ldots,n$, define
\[
C_k^-:=\frac1n\sum_{i=1}^k \cos(\theta_{(i)}), \qquad
S_k^-:=\frac1n\sum_{i=1}^k \sin(\theta_{(i)}),
\]
and
\[
C_\alpha^{(k)}:=\alpha C_n+(1-2\alpha)C_k^-, \qquad
S_\alpha^{(k)}:=(1-2\alpha)S_k^-.
\]
Then, for $\alpha\in(0,1/2)$ and $\nu\in(\theta_{(k)},\theta_{(k+1)})$,
\[
I_{n,\alpha}(\nu)=C_\alpha^{(k)}\sin\nu - S_\alpha^{(k)}\cos\nu.
\]
Moreover, the empirical circular $\alpha$-expectile is uniquely defined. There exists a unique zero $q_n(\alpha)\in(-\pi,0)$ of $I_{n,\alpha}$, and this zero is the unique minimizer of $g_{n,\alpha}(\nu,\theta)$.

Let $k_*$ be the unique index such that $-\pi/2\in(\theta_{(k_*)},\theta_{(k_*+1)})$. 
Then exactly one of the following holds:
\begin{align*}
C_\alpha^{(k_*)}<0  \; &\Rightarrow \; q_n(\alpha)\in(-\pi,-\pi/2), \\
C_\alpha^{(k_*)}=0  \; &\Rightarrow \; q_n(\alpha)=-\pi/2, \\
C_\alpha^{(k_*)}>0  \; &\Rightarrow \; q_n(\alpha)\in(-\pi/2,0).
\end{align*}
\end{theorem}

\begin{remark}[Empirical upper circular expectiles]
The case $\alpha\in(1/2,1)$ follows by reflection. After centering at $\hat\mu_n=0$, let $\eta_i=-\theta_i$. Then
\[
I_{n,\alpha}(\nu;\theta) = - I_{n,1-\alpha}(-\nu;\eta).
\]
Hence the equation $I_{n,\alpha}(\nu;\theta)=0$ has a unique solution $q_n(\alpha)\in(0,\pi)$, and this solution is the unique minimizer of $g_{n,\alpha}$.
\end{remark}

\begin{remark}
Let $q_n(\alpha)$ denote the empirical circular $\alpha$-expectile after centering the
sample so that $\hat\mu_n=0$, and let $-\pi<\theta_{(1)}<\cdots<\theta_{(n)}<\pi$ be the ordered sample. Then
\[
\lim_{\alpha\downarrow 0} q_n(\alpha)=\theta_{(1)}, \qquad
\lim_{\alpha\uparrow 1} q_n(\alpha)=\theta_{(n)}.
\]
This follows from the same sign-change argument: as $\alpha\downarrow 0$, the unique zero is forced into the first interval $(\theta_{(1)},\theta_{(2)})$ and converges to $\theta_{(1)}$; the upper
limit follows by reflection.
Thus, in the empirical case, the extreme circular expectiles converge to the extreme
sample points in the linearized order.
\end{remark}

\begin{remark}[Computation of the empirical circular expectile]
\label{rem:empirical-algorithm}
After rotating the sample so that $\hat\mu_n=0$, let
$-\pi=\theta_{(0)}<\theta_{(1)}<\cdots<\theta_{(n)}<\theta_{(n+1)}=\pi$ be the ordered sample. On each interval $(\theta_{(k)},\theta_{(k+1)})$,
\[
I_{n,\alpha}(\nu) = C_{\alpha}^{(k)}\sin\nu - S_{\alpha}^{(k)}\cos\nu .
\]
For $k=1,\ldots,n-1$, one has $S_\alpha^{(k)}<0$, and hence
\[
\phi_k = \operatorname{atan2}\!\left(S_\alpha^{(k)},C_\alpha^{(k)}\right) \in(-\pi,0).
\]
The zeros of $I_{n,\alpha}$ on the interval $(\theta_{(k)},\theta_{(k+1)})$ are the points satisfying
\[
\nu\equiv \phi_k \pmod{\pi}.
\]
By Theorem~\ref{prop:empirical}, for $\alpha\in(0,1/2)$ there is a unique zero in $(- \pi,0)$. Hence there is a unique index $k$ such that
\[
\phi_k\in(\theta_{(k)},\theta_{(k+1)}),
\]
and this value is the empirical circular $\alpha$-expectile. The upper case $\alpha\in(1/2,1)$ follows by reflection. The resulting algorithm has complexity $O(n)$ for already ordered data.
\end{remark}

Figure \ref{fig1} shows the empirical lower and upper circular $\alpha$-expectiles for several values of $\alpha$ for the von Mises distribution $vM(0,\kappa)$ with $\kappa=3$ and $\kappa=1$.

\begin{figure}
\centering
\includegraphics[width=0.9\textwidth]{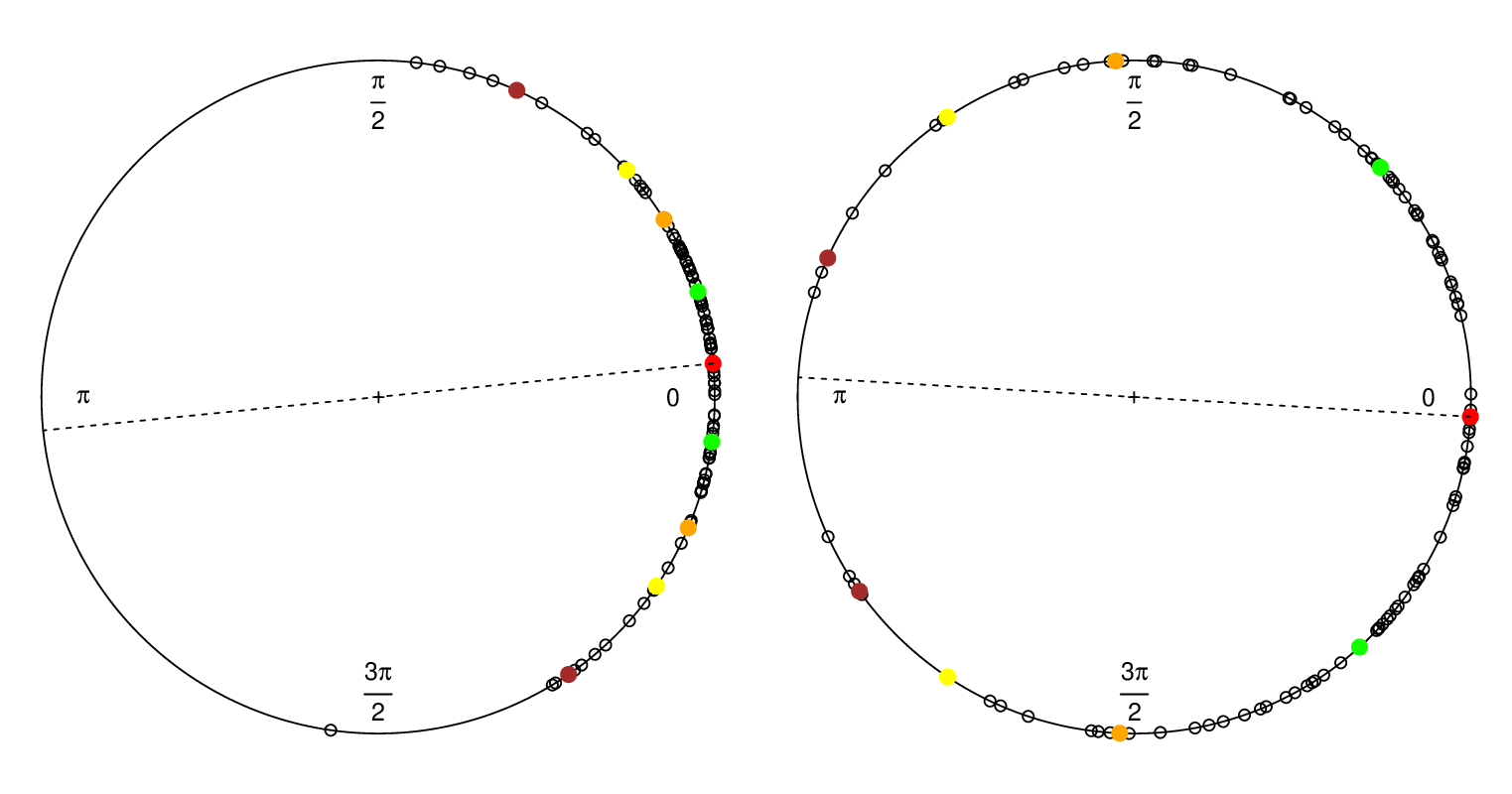}
\caption{Empirical lower and upper circular $\alpha$-expectiles for $\alpha=0.5$ (red), $\alpha=0.25$ (green), 0.1 (orange), 0.05 (yellow), 0.01 (brown) for samples of size 100 from the von Mises distribution $vM(0,\kappa)$ with $\kappa=3$ (left) and $\kappa=1$ (right).} \label{fig1}
\end{figure}


\section{Asymptotic properties of empirical circular expectiles}

First, we present a consistency result for the empirical  circular $\alpha$-expectile. The proof, which is deferred to the appendix, employs a sign-change argument similar to that in Lemma 5.10 of \cite{VV:1998}. \cite{HK:2016} use the same sign-based consistency argument for linear expectiles. Additionally, it must be shown that replacing the true branch center $\mu$ by $\hat\mu_n$ does not affect the empirical identification function asymptotically.

\begin{proposition}[Consistency of the empirical circular expectile] \label{consistency}
Let $X_1,X_2,\ldots$ be iid from an absolutely continuous circular
distribution. Assume $R:=|\Ex e^{iX}|>0$, and let $\mu$ be the unique circular mean.
Let $\alpha\in(0,1/2)$, and suppose that the population circular
$\alpha$-expectile $\mu_\alpha\in(\mu-\pi,\mu+\pi)$ is the unique zero of
$I_\alpha(\nu,F)$, with sign change from negative to positive.
Let $\hat\mu_{n,\alpha}$ denote the empirical circular $\alpha$-expectile
defined in Section~3. Define its $\mu$-centered representative by
\[
\hat\mu_{n,\alpha}^{(\mu)} := \mu+\operatorname{Arg}_{(-\pi,\pi)}
\left(e^{i(\hat\mu_{n,\alpha}-\mu)}\right).
\]
Then
\[
\hat\mu_{n,\alpha}^{(\mu)}\to\mu_\alpha \quad\text{a.s.}
\]
\end{proposition}

\begin{theorem}[Finite-dimensional CLT for empirical circular expectiles] \label{thm:clt}
Let $X_1,X_2,\ldots$ be iid from an absolutely continuous circular distribution. Assume 
$R:=|\Ex e^{iX}|>0$, and let $\mu$ be the unique circular mean. Let 
\[
Y := \operatorname{Arg}_{(-\pi,\pi)} \left(e^{i(X-\mu)}\right)
\]
be the $\mu$-centered representative of $X$. Thus $Y\in(-\pi,\pi)$ a.s., $\Ex\sin Y=0$,
$R=\Ex\cos Y>0$. Assume that $Y$ has a density $f$ which is positive a.e. on
$(-\pi,\pi)$. Assume further that $f$ is continuous at the antipode of
$0$, and write
\[
f_A := \lim_{y\downarrow-\pi}f(y) = \lim_{y\uparrow\pi}f(y).
\]
Let $\alpha_1,\ldots,\alpha_k\in(0,1)$, and let $q_j, j=1,\ldots,k,$
denote the corresponding population circular expectiles in the $\mu$-centered chart. For $j=1,\ldots,k$, define
\[
I_{\alpha_j}(q_j,y) = \alpha_j\sin(q_j-y)\mathbf 1_{\{y>q_j\}}
+ (1-\alpha_j)\sin(q_j-y)\mathbf 1_{\{y\le q_j\}},
\quad 
A_j = \Big. \frac{\partial}{\partial\nu}I_{\alpha_j}(\nu,F) \Big|_{\nu=q_j}.
\]
Then $A_j>0$. Equivalently, $A_j=C_{\alpha_j}(q_j)\cos q_j+S_{\alpha_j}(q_j)\sin q_j$. Set
\[
B_j = (1-2\alpha_j)\sin q_j\, f_A .
\]
Let $\hat\mu_{n,\alpha_j}$ be the empirical circular $\alpha_j$-expectile defined in Section~3, and let
$\hat q_{n,j}$ be its $\mu$-centered representative. Then
\[
\sqrt n \begin{pmatrix}
\hat q_{n,1}-q_1\\ \vdots\\ \hat q_{n,k}-q_k
\end{pmatrix}
\; \stackrel{\cal L}{\longrightarrow} \;  N_k(0,\Sigma),
\]
where
\[
\Sigma_{j\ell} = \frac{\Ex[W_j(Y)W_\ell(Y)]}{A_jA_\ell}, \quad j,\ell=1,\ldots,k,
\]
with
\[
W_j(Y) = I_{\alpha_j}(q_j,Y) + \frac{B_j}{R}\sin Y.
\]
\end{theorem}

\begin{proof}
By rotation invariance, we work in the $\mu$-centered chart and assume $\mu=0$. Thus $Y=X\in(-\pi,\pi)$ a.s., $\Ex\sin Y=0$, $R=\Ex\cos Y>0$. Let $M_n$ 
denote the sample circular mean in this chart. Then $M_n\to 0$ a.s.

By Theorem~3 and Remark~4, for $\alpha_j\in(0,1), \alpha_j\neq 1/2$, the population circular $\alpha_j$-expectile $q_j$ is the unique sign-changing zero of $I_{\alpha_j}(\nu,F)$, and $A_j>0$ holds.
For $\alpha_j=1/2$, one has $q_j=0$, $I_{1/2}(\nu,F)=R\sin(\nu)/2, A_j=R/2>0$.
Moreover, by Theorem~8 and Remark~9, the empirical circular expectiles are uniquely defined almost surely. Proposition~12, together with reflection for $\alpha_j>1/2$ and the usual consistency of the sample circular mean for $\alpha_j=1/2$, gives $\hat q_{n,j}\to q_j$ in probability for $j=1,\ldots,k$.

For a branch center $m$, let
\[
y^{(m)} = m+\operatorname{Arg}_{(-\pi,\pi)} \left(e^{i(y-m)}\right)
\]
denote the representative of $y$ in $(m-\pi,m+\pi)$. For $\nu\in(m-\pi,m+\pi)$, set
\[
\psi_\alpha(m,\nu;y) = \alpha\sin(\nu-y^{(m)})\mathbf 1_{\{y^{(m)}>\nu\}}
+ (1-\alpha)\sin(\nu-y^{(m)})\mathbf 1_{\{y^{(m)}\le\nu\}} .
\]
The sample circular mean satisfies
\[
P_n\sin(M_n-Y)=0 .
\]
For the empirical circular expectile, let $\tilde q_{n,j}\in(M_n-\pi,M_n+\pi)$ denote the representative used in the empirical criterion. By the first-order condition for the empirical expectile,
\[
P_n\psi_{\alpha_j}(M_n,\tilde q_{n,j};Y)=0,
\qquad j=1,\ldots,k .
\]
Since $M_n\to0$ and $\hat q_{n,j}\to q_j\in(-\pi,\pi)$, the
$\mu$-centered representative $\hat q_{n,j}$ belongs to
$(M_n-\pi,M_n+\pi)$ with probability tending to one, and hence $\tilde q_{n,j}=\hat q_{n,j}$
with probability tending to one. Therefore, in the local asymptotic argument, we may write the empirical equations as
\[
P_n\sin(M_n-Y)=0, \qquad P_n\psi_{\alpha_j}(M_n,\hat q_{n,j};Y)=0, \quad j=1,\ldots,k,
\]
with probability tending to one. Here $P_n$ denotes the empirical measure of $Y_1,\ldots,Y_n$.
Define $\hat\theta_n=(M_n,\hat q_{n,1},\ldots,\hat q_{n,k})$ and the vector-valued estimating function
\[
\Psi_n(m,\nu_1,\ldots,\nu_k) =  P_n
\begin{pmatrix}
\sin(m-Y)\\ \psi_{\alpha_1}(m,\nu_1;Y)\\ \vdots\\ \psi_{\alpha_k}(m,\nu_k;Y)
\end{pmatrix}.
\]
Then the preceding equations imply 
\[
\Psi_n(\hat\theta_n)=0
\]
with probability tending to one. In particular, $\sqrt n\,\Psi_n(\hat\theta_n)=o_P(1)$,
which is the approxi\-mate-zero condition in the Z-estimator expansion \citep[Th.~5.21]{VV:1998}.
Let
\[
\Psi(m,\nu_1,\ldots,\nu_k) = \Ex\Psi_n(m,\nu_1,\ldots,\nu_k),
\]
and put $\theta_0:=(0,q_1,\ldots,q_k)$. Then $\Psi(\theta_0)=0$. Moreover,
\[
\sqrt n\,\Psi_n(\theta_0) = \frac1{\sqrt n}\sum_{i=1}^n
\begin{pmatrix}
-\sin Y_i\\ I_{\alpha_1}(q_1,Y_i)\\ \vdots\\ I_{\alpha_k}(q_k,Y_i)
\end{pmatrix}.
\]
The summands are centered and bounded; in particular, the finite second-moment condition required in the Z-estimator argument is satisfied.

We next compute the derivative matrix of $\Psi$ at $\theta_0$, obtaining 
\[
\left. \frac{\partial}{\partial m} \Ex\sin(m-Y) \right|_{m=0} = \Ex\cos Y = R, 
\quad \text{and} \quad 
\left. \frac{\partial}{\partial\nu} \Ex\psi_{\alpha_j}(0,\nu;Y) \right|_{\nu=q_j} = A_j, 
\]
for $j=1,\ldots,k$. 
It remains to compute the derivative of $\Ex\psi_{\alpha_j}(m,q_j;Y)$ with respect to the branch center $m$. This derivative is caused only by observations crossing the branch cut at the antipode. For $h>0$, the observations in $(-\pi,-\pi+h]$ are moved from the left end of the interval to the right end. Hence
\[
\Ex\psi_{\alpha_j}(h,q_j;Y)-\Ex\psi_{\alpha_j}(0,q_j;Y)
= (2\alpha_j-1) \int_{-\pi}^{-\pi+h} \sin(q_j-y)f(y)\,dy.
\]
Dividing by $h$ and letting $h\downarrow0$, using the continuity of $f$ at the antipode, gives
\[
(2\alpha_j-1)\sin(q_j+\pi)f_A = (1-2\alpha_j)\sin q_j\, f_A.
\]
The same derivative is obtained from the left by considering the interval
$(\pi+h,\pi)$ for $h<0$. Therefore
\[
\left. \frac{\partial}{\partial m} \Ex\psi_{\alpha_j}(m,q_j;Y) \right|_{m=0} 
= B_j := (1-2\alpha_j)\sin q_j\,f_A .
\]
Thus
\[
D\Psi(\theta_0) =
\begin{pmatrix}
R & 0 & 0 & \cdots & 0\\
B_1 & A_1 & 0 & \cdots & 0\\
B_2 & 0 & A_2 & \cdots & 0\\
\vdots & \vdots & \vdots & \ddots & \vdots\\
B_k & 0 & 0 & \cdots & A_k
\end{pmatrix}.
\]
This matrix is nonsingular because $R>0$ and $A_j>0$ for all $j=1,\ldots,k$.

It remains to verify the stochastic equicontinuity condition used in the finite-dimensional Z-estimator argument; see Theorem~5.21 and the remarks following it in \citet{VV:1998}. Let
\[
\theta=(m,\nu_1,\ldots,\nu_k), \quad \theta_0=(0,q_1,\ldots,q_k),
\]
and write $h_\theta$ for the vector-valued estimating function appearing inside $P_n$ in the definition of $\Psi_n$. 
The class $\{h_\theta:\theta$ in a neighbourhood of $\theta_0\}$ is bounded and of VC type, since its components are bounded trigonometric functions multiplied by indicators of intervals, or finite unions of intervals, with endpoints depending on finitely many parameters. Hence the class is $P$-Donsker.

Moreover, $h_\theta$ is $L_2(P)$-continuous at $\theta_0$. The first component satisfies
\[
|\sin(m-y)-\sin(-y)|\le |m|.
\]
For one expectile component, for $(m,\nu)$ sufficiently close to $(0,q_j)$, there is a constant $C>0$ such that
\[ 
\left| \psi_{\alpha_j}(m,\nu;y)-\psi_{\alpha_j}(0,q_j;y) \right|
\le C|\nu-q_j| + C\mathbf 1_{\{|y-q_j|\le C|\nu-q_j|\}} + C\mathbf 1_{A_{|m|}}(y),
\]
where $A_\delta=(-\pi,-\pi+\delta]\cup[\pi-\delta,\pi)$. Since $Y$ has a density,
\[
P(|Y-q_j|\le\delta)\to 0, \quad P(Y\in A_\delta)\to 0 \quad(\delta\downarrow0).
\]
Therefore
\[
\|h_\theta-h_{\theta_0}\|_{L_2(P)}\to 0 \quad(\theta\to\theta_0).
\]
Since the class is $P$-Donsker, $h_\theta$ is $L_2(P)$-continuous at $\theta_0$,  and 
$\hat\theta_n\to\theta_0$ in probability, Lemma~19.24 of \citet{VV:1998} implies
\[
\sqrt n(P_n-P)(h_{\hat\theta_n}-h_{\theta_0})=o_P(1),
\]
which corresponds to (5.22) in the proof of Theorem~5.21 of \citet{VV:1998}; see also the remarks following that theorem. 
Together with the approximate-zero condition, differentiability of $\Psi$ at $\theta_0$, and nonsingularity of $D\Psi(\theta_0)$, the finite-dimensional Z-estimator expansion applies:
\[
\sqrt n 
\begin{pmatrix}
M_n\\ \hat q_{n,1}-q_1\\ \vdots\\ \hat q_{n,k}-q_k
\end{pmatrix}
= - \left(D\Psi(\theta_0)\right)^{-1} \frac1{\sqrt n} \sum_{i=1}^n
\begin{pmatrix}
-\sin Y_i\\ I_{\alpha_1}(q_1,Y_i)\\ \vdots\\ I_{\alpha_k}(q_k,Y_i)
\end{pmatrix}
+ o_P(1).
\]
The first row gives
\[
\sqrt n\,M_n = \frac1R \frac1{\sqrt n}\sum_{i=1}^n \sin Y_i + o_P(1).
\]
For $j=1,\ldots,k$, the $j$-th expectile row gives
\[
B_j\sqrt n\,M_n + A_j\sqrt n(\hat q_{n,j}-q_j) + \frac1{\sqrt n}\sum_{i=1}^n I_{\alpha_j}(q_j,Y_i)
= o_P(1).
\]
Substituting the expansion of $\sqrt n\,M_n$, we obtain
\[
\sqrt n(\hat q_{n,j}-q_j) = -\frac1{A_j} \frac1{\sqrt n}\sum_{i=1}^n 
\left[ I_{\alpha_j}(q_j,Y_i) + \frac{B_j}{R}\sin Y_i \right] + o_P(1).
\]
With
\[
W_j(Y) = I_{\alpha_j}(q_j,Y) + \frac{B_j}{R}\sin Y,
\]
this yields the joint linearization
\[
\sqrt n
\begin{pmatrix}
\hat q_{n,1}-q_1\\ \vdots\\ \hat q_{n,k}-q_k
\end{pmatrix}
= 
- \frac1{\sqrt n}\sum_{i=1}^n
\begin{pmatrix}
W_1(Y_i)/A_1\\ \vdots\\ W_k(Y_i)/A_k 
\end{pmatrix}
+ o_P(1).
\]
Since
\[
\Ex W_j(Y) = \Ex I_{\alpha_j}(q_j,Y) + \frac{B_j}{R}\Ex\sin Y = 0,
\]
and the vector $(W_1(Y),\ldots,W_k(Y))$ is bounded, the multivariate central
limit theorem gives
\[
\sqrt n
\begin{pmatrix}
\hat q_{n,1}-q_1\\ \vdots\\ \hat q_{n,k}-q_k 
\end{pmatrix}
\stackrel{\cal{L}}{\longrightarrow} N_k(0,\Sigma),
\]
where
\[
\Sigma_{j\ell} = \frac{\Ex[W_j(Y)W_\ell(Y)]}{A_jA_\ell}, \quad j,\ell=1,\ldots,k.
\]
This proves the assertion.
\end{proof}

\begin{remark}
For $\alpha_j=1/2$, the corresponding expectile is the circular mean. In this case $q_j=0$,
\[
I_{1/2}(0,y)=-\frac12\sin y,\quad A_j=\frac R2,\quad B_j=0.
\]
Hence Theorem~13 yields
\[
\sqrt n\,\hat q_{n,j}
=
\frac1R\frac1{\sqrt n}\sum_{i=1}^n\sin Y_i+o_P(1),
\]
and therefore
\[
\sqrt n\,\hat q_{n,j} \stackrel{\cal L}{\longrightarrow} N\left(0,\frac{\Ex[\sin^2Y]}{R^2}\right).
\]
Thus the theorem contains the classical asymptotic normality of the sample
circular mean as the special case $\alpha=1/2$.    
\end{remark}

\begin{remark}
The assumptions of Theorem~\ref{thm:clt} are stronger than necessary. For the asymptotic normality result it is enough to assume uniqueness of the relevant population roots, positivity of the corresponding slopes $A_j$, consistency of the empirical roots, and continuity of the density at the
antipode. We use the assumptions of Theorem~\ref{thm:main-lower-expectile} in order to keep the statement aligned with the preceding uniqueness theory.
\end{remark}

Fig. \ref{fig2} shows stack plots of the circular mean and the circular $\alpha$-expectiles for $\alpha=0.75$ and 0.9 for samples of sizes 25, 100 and 400 from the von Mises distribution $vM(0,1)$, based on $10^4$ replications.

\begin{figure}
\centering
\includegraphics[width=\textwidth]{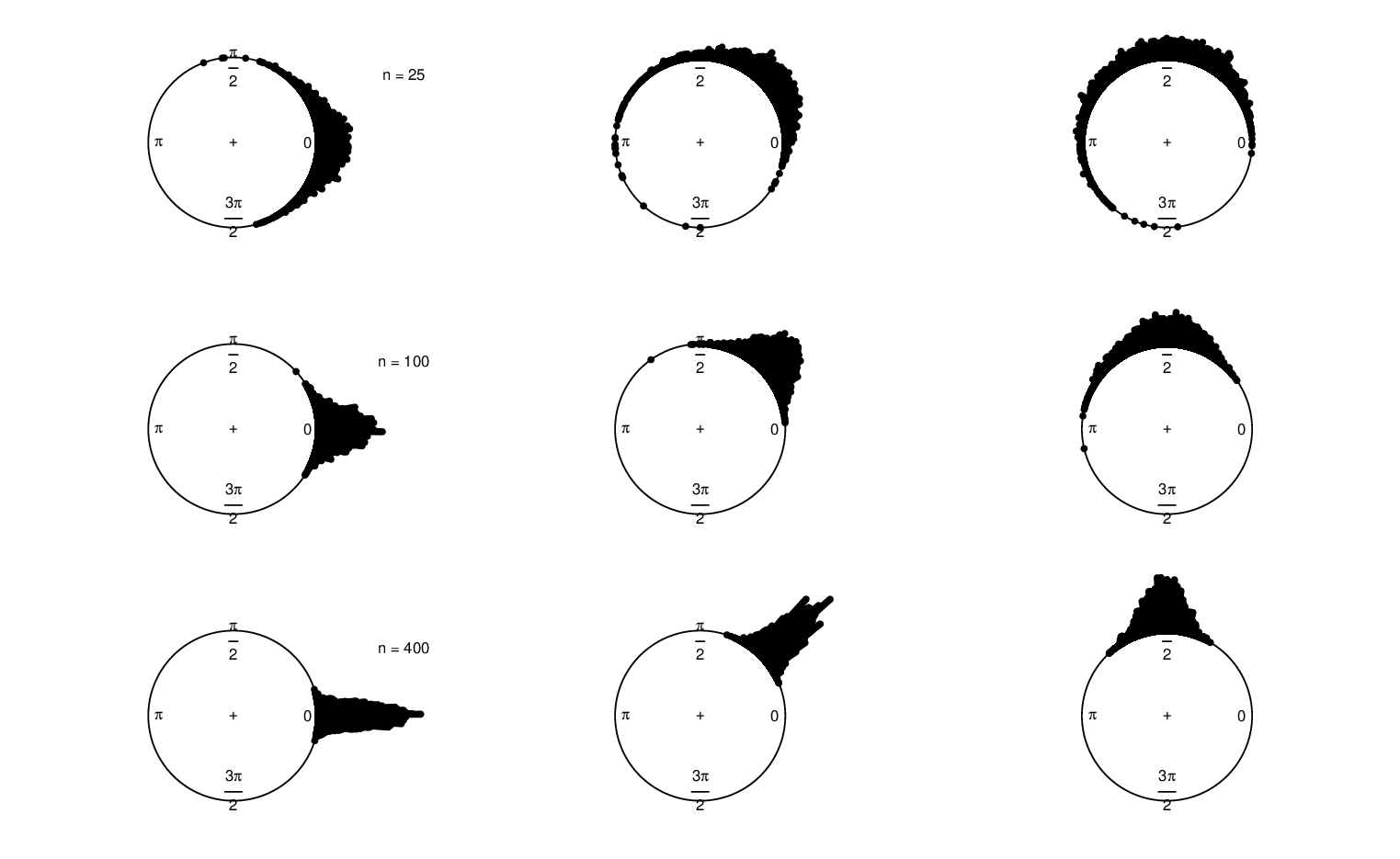}
\caption{Stack plot of $\hat\mu_{n}$ (left), $\hat\mu_{n,0.75}$ (middle) and $\hat\mu_{n,0.90}$ (right) for samples of sizes 25 (upper row), 100 (middle row) and 400 (lower row) from the von Mises distribution $vM(0,1)$, with $10^4$ replications.} \label{fig2}
\end{figure}

As a numerical check of Theorem~13, Table~\ref{tab:clt-simulation} compares the limiting covariance matrix for the pair $(\alpha_1,\alpha_2)=(0.75,0.90)$ with Monte Carlo covariance matrices of
\[
\sqrt n(\hat q_{n,.75}-q_{.75},\hat q_{n,.90}-q_{.90})
\]
for samples from $vM(0,1)$. The approximation is already reasonable for $n=25$, especially for the correlation, and becomes accurate for $n=100$ and $n=400$. This supports the finite-dimensional normal
approximation derived above.

All computations were carried out in R \citep{R:2025}, using the
\texttt{circular} package \citep{CI:2025}.

Code for computing theoretical and empirical circular expectiles can be found on GitHub at \url{https://github.com/BernhardKlar/CircularExpectiles}

\begin{table}[t]
\centering
\begin{tabular}{c|rrrr}
\hline
 & $\operatorname{Var}_{.75}$ & $\operatorname{Cov}_{.75,.90}$
 & $\operatorname{Var}_{.90}$ & $\operatorname{Corr}_{.75,.90}$ \\ \hline
$n=25$  & 4.84 & 5.76 & 10.08 & 0.825\\
$n=100$ & 4.15 & 5.84 & 11.96 & 0.829\\
$n=400$ & 4.03 & 5.92 & 12.57 & 0.831\\
Limit   & 3.94 & 5.83 & 12.53 & 0.829\\
\hline
\end{tabular}
\caption{Monte Carlo covariance matrices of 
$\sqrt n(\hat q_{n,.75}-q_{.75},\hat q_{n,.90}-q_{.90})$ for samples from the von Mises distribution $vM(0,1)$, based on $10^5$ replications, compared with the limiting covariance matrix from Theorem~13.}
\label{tab:clt-simulation}
\end{table}

\subsection{Potential applications}

The results developed above suggest several possible applications of circular expectiles. Throughout this subsection, angular differences are understood in the $\mu$-centered chart $(\mu-\pi,\mu+\pi)$. Thus, for $\alpha\in(0,1/2)$, the lower expectile $\mu_\alpha$ lies to the left of $\mu$, while the upper expectile $\mu_{1-\alpha}$ lies to the right of $\mu$, and expressions such as 
$\mu_{1-\alpha}-\mu_\alpha$ are interpreted as ordinary differences of these representatives.

First, circular expectiles lead naturally to expectile-based measures of dispersion. In analogy with the linear case \citep{BFW:2022,EK:2023}, one may define the circular interexpectile range by
\[
R_\alpha = \mu_{1-\alpha}-\mu_\alpha, \quad \alpha\in(0,1/2).
\]
This quantity measures the width of the central expectile interval around the circular mean direction. It is small for concentrated distributions and becomes large when the distribution spreads over a substantial part of the circle. The empirical counterpart is 
$\hat R_{n,\alpha} = \hat\mu_{n,1-\alpha}-\hat\mu_{n,\alpha}$. The finite-dimensional
asymptotic normality established above provides the basis for studying the
large-sample distribution of $\hat R_{n,\alpha}$, for constructing confidence
intervals, and for comparing such expectile-based dispersion measures with
classical circular measures of spread.

Second, the circular expectile curve carries information about the full distribution. Indeed, under regularity assumptions such as those used in the uniqueness results above, the full map
$\alpha\mapsto\mu_\alpha,\, \alpha\in(0,1)$, determines the underlying circular distribution. To see this informally, assume
by rotation invariance that $\mu=0$. For $0<\alpha<1/2$, the expectile $q(\alpha)=\mu_\alpha\in(-\pi,0)$ satisfies
\[
0 = I_\alpha(q(\alpha),F) = \alpha C\sin q(\alpha) + (1-2\alpha)D(q(\alpha)),
\]
where 
\[
C=\Ex\cos X>0, \qquad D(\nu) = \int_{-\pi}^{\nu}\sin(\nu-x)\,dF(x).
\]
Since $q$ is strictly monotone, the lower expectile curve determines
$D(\nu)/C$ on $(-\pi,0)$. The upper expectile curve gives the analogous
information on $(0,\pi)$. If $F$ has density $f$, then
\[
D''(\nu)+D(\nu)=f(\nu),
\]
so the expectile curve determines the density, up to the normalizing constant,
which is fixed by $\int_{-\pi}^{\pi}f(x)\,dx=1$.

This identifiability observation can be used to characterize symmetry. A circular distribution is symmetric about $\mu$ if and only if
\[ 
\mu_{1-\alpha}-\mu = \mu-\mu_\alpha \quad \text{for all } \alpha\in(0,1/2).
\]
Thus circular expectiles provide a natural basis for graphical diagnostics and formal tests of circular symmetry. For example, one may plot the two distances $\mu_{1-\alpha}-\mu$ and $\mu-\mu_\alpha$ as functions of $\alpha$, or plot their difference. This parallels expectile-based symmetry diagnostics in
the linear case \citep{EK:2022}.

Third, circular expectiles can be used to define measures of skewness. In
analogy with the expectile-based skewness measures studied in the linear
setting \citep{EK:2022,EK:2023}, define
\[
\gamma_\alpha = \frac{\mu_{1-\alpha}+\mu_\alpha-2\mu}{\mu_{1-\alpha}-\mu_\alpha}
= \frac{d_\alpha^+-d_\alpha^-}{d_\alpha^++d_\alpha^-}, \quad \alpha\in(0,1/2),
\]
where $d_\alpha^-=\mu-\mu_\alpha$ and $d_\alpha^+=\mu_{1-\alpha}-\mu$.
Hence $\gamma_\alpha\in(-1,1)$, with $\gamma_\alpha=0$ corresponding to balanced lower and upper expectile distances at level $\alpha$. Positive values indicate a longer upper expectile tail, while negative values indicate a longer lower expectile tail. The function $\alpha\mapsto\gamma_\alpha$ therefore yields a skewness profile of the circular distribution. Its empirical counterpart,
\[
\hat\gamma_{n,\alpha} 
= \frac{\hat\mu_{n,1-\alpha}+\hat\mu_{n,\alpha}-2\hat\mu_n}{\hat\mu_{n,1-\alpha}-\hat\mu_{n,\alpha}},
\]
can be studied by applying the delta method to the joint asymptotic distribution of
$(\hat\mu_n,\hat\mu_{n,\alpha},\hat\mu_{n,1-\alpha})$.
This opens the way to comparing expectile-based circular skewness measures with existing measures of circular skewness, and to studying their finite-sample and asymptotic behaviour; cf. \citet{EK:2020} for related investigations in the linear case.

Several further directions are possible. Since the full circular expectile curve identifies the distribution, one may develop goodness-of-fit procedures based on distances between empirical and fitted expectile curves. Finally, another natural direction is conditional or regression-type circular expectiles, where the center, lower expectiles and upper expectiles are allowed to depend on
covariates.

\appendix
\section{Proofs of Theorem~\ref{prop:empirical} and Proposition~\ref{consistency}}    

\begin{proof}[Proof of Theorem~\ref{prop:empirical}]
For $\nu\in(\theta_{(k)},\theta_{(k+1)})$, the indicator sets are constant, so
\[
C_\nu^- := \frac1n \sum_{i=1}^n \cos(\theta_i)\mathbf 1_{\{\theta_i\le \nu\}},
\quad
S_\nu^- := \frac1n \sum_{i=1}^n \sin(\theta_i)\mathbf 1_{\{\theta_i\le \nu\}},
\]
satisfy $C_\nu^- = C_k^-$ and $S_\nu^- = S_k^-$. Hence
\[
I_{n,\alpha}(\nu) = \alpha C_n\sin\nu + (1-2\alpha)\bigl(C_k^-\sin\nu-S_k^-\cos\nu\bigr)
= C_\alpha^{(k)}\sin\nu - S_\alpha^{(k)}\cos\nu.
\]
We next analyze the coefficients. The sequence $(S_k^-)_{k=0}^n$ is strictly decreasing as long as $\theta_{(k)}<0$, strictly increasing after the sign change, and satisfies
$S_0^-=S_n^-=0$. Hence $S_k^-<0$ for all $k=1,\ldots,n-1$. Therefore, for $\alpha\in(0,1/2)$,
\[
S_\alpha^{(k)}=(1-2\alpha)S_k^-<0 \qquad \text{for all } k=1,\ldots,n-1.
\]
For $k=0$ and $k=n$, one has
\[
S_\alpha^{(0)}=S_\alpha^{(n)}=0, \qquad
C_\alpha^{(0)}=\alpha C_n>0, \qquad C_\alpha^{(n)}=(1-\alpha)C_n>0.
\]
Therefore, 
\begin{align} \label{C-S-not-0}
(C_\alpha^{(k)},S_\alpha^{(k)}) \neq (0,0) \quad \text{for all } k=0,\ldots,n.
\end{align}
Hence, on each interval $(\theta_{(k)},\theta_{(k+1)})$, the function
\[
I_{n,\alpha}(\nu)=C_\alpha^{(k)}\sin\nu - S_\alpha^{(k)}\cos\nu
\]
is a nontrivial sinusoid. Now define, on $(-\pi/2,\pi/2)$,
\[
J_{n,\alpha}(\nu):=\frac{I_{n,\alpha}(\nu)}{\cos\nu}.
\]
Since each summand in $I_{n,\alpha}$ vanishes at a sample point $\theta_i$, the function $I_{n,\alpha}$ is continuous on $(-\pi,\pi)$; hence $J_{n,\alpha}$ is continuous on $(-\pi/2,\pi/2)$. On each interval $(\theta_{(k)},\theta_{(k+1)})\subset(-\pi/2,\pi/2)$, differentiation yields
\[
J'_{n,\alpha}(\nu)=\frac{C_\alpha^{(k)}}{\cos^2\nu}.
\]
We distinguish three cases according to the sign of $C_\alpha^{(k_*)}$, where $k_*$ is the unique index such that $-\pi/2 \in (\theta_{(k_*)},\theta_{(k_*+1)})$.

\medskip\noindent
\emph{Case 1: $C_\alpha^{(k_*)}>0$.}

\begin{sloppypar}
Since
\[
C_\alpha^{(k+1)}-C_\alpha^{(k)}=\frac{1-2\alpha}{n}\cos(\theta_{(k+1)}),
\]
the sequence $k\mapsto C_\alpha^{(k)}$ is strictly increasing on the index range corresponding to 
$(-\pi/2,\pi/2)$. Hence $C_\alpha^{(k)}>0$ for every $k$ such that $(\theta_{(k)},\theta_{(k+1)})\subset(-\pi/2,\pi/2)$, and therefore $J_{n,\alpha}$ is strictly increasing on $(-\pi/2,\pi/2)$. Moreover,
\[
\lim_{\nu\downarrow-\pi/2}J_{n,\alpha}(\nu)=-\infty,
\qquad
J_{n,\alpha}(0)=-S_\alpha^{(k_0)}>0,
\]
where $k_0$ is the unique index such that $0\in(\theta_{(k_0)},\theta_{(k_0+1)})$. Thus there exists exactly one zero in $(-\pi/2,0)$ and no zero in $(0,\pi/2)$.
\end{sloppypar}

On $(-\pi,-\pi/2)$, one has $\sin\nu<0$, $\cos\nu<0$, $S_\alpha^{(k)}<0$, and $C_\alpha^{(k)}\ge 0$ up to the interval containing $-\pi/2$, so
\[
I_{n,\alpha}(\nu)=C_\alpha^{(k)}\sin\nu-S_\alpha^{(k)}\cos\nu<0.
\]

Finally, on $(\pi/2,\pi)$ the coefficients satisfy $C_{\alpha}^{(k)}>0$ and $S_{n,\alpha}^{(k)}\le 0$. Indeed, $C_{\alpha}^{(k)}$ is decreasing on the right tail and ends at
\[
C_{\alpha}^{(n)}=(1-\alpha)C_n>0,
\]
while $S_{\alpha}^{(k)}<0$ for $k=1,\ldots,n-1$ and $S_{\alpha}^{(n)}=0$. 
Hence, for $\nu\in(\pi/2,\pi)$,
\[
I'_{n,\alpha}(\nu) = C_{\alpha}^{(k)}\cos\nu + S_{\alpha}^{(k)}\sin\nu <0 .
\]
Since $\lim_{\nu\uparrow\pi} I_{n,\alpha}(\nu)=0$, it follows that $I_{n,\alpha}(\nu)>0$ for all $\nu\in(\pi/2,\pi)$.
Hence the zero is unique.

\medskip\noindent
\emph{Case 2: $C_\alpha^{(k_*)}=0$.}

Then $I_{n,\alpha}(-\pi/2)=-C_\alpha^{(k_*)}=0$.
As in Case~1, $J_{n,\alpha}$ is increasing on $(-\pi/2,\pi/2)$. 
Note that on the interval containing $-\pi/2$, one has $C_\alpha^{(k_*)}=0$, so
$J_{n,\alpha}'(\nu)=0$, that is, $J_{n,\alpha}$ is constant there. However,
\[
I_{n,\alpha}(\nu) = -S_\alpha^{(k_*)}\cos\nu,
\]
and since $-S_\alpha^{(k_*)}>0$, it follows that $I_{n,\alpha}$ is strictly increasing on this interval. It follows that 
\[
I_{n,\alpha}(\nu)>0 \quad \text{for all }\nu\in(-\pi/2,\pi/2).
\]
Again as in Case~1, one has $I_{n,\alpha}(\nu)<0$ on $(-\pi,-\pi/2)$ and
$I_{n,\alpha}(\nu)>0$ on $(\pi/2,\pi)$. Hence the unique zero is $\nu=-\pi/2$.

\medskip\noindent
\emph{Case 3: $C_\alpha^{(k_*)}<0$.}

Since $C_{\alpha}^{(0)}>0$ and $C_{\alpha}^{(k_*)}<0$, define
\[
m:=\min\{k\le k_*: C_{\alpha}^{(k)}<0\}.
\]
Then $m\ge1$, $C_{\alpha}^{(m-1)}\ge0$, and $\theta_{(m)}<-\pi/2$. Moreover,
$C_{\alpha}^{(k)}<0$ for all $k=m,\ldots,k_*$, since $k\mapsto C_{\alpha}^{(k)}$ is strictly decreasing on the left tail $(-\pi,-\pi/2)$.

First consider $(-\pi,-\pi/2)$. For $k<m$, we have $C_{\alpha}^{(k)}\ge0$, $S_{\alpha}^{(k)}\le0$, and, on this interval, $\sin\nu<0$, $\cos\nu<0$. Hence, by \eqref{C-S-not-0},
\[
I_{n,\alpha}(\nu) = C_{\alpha}^{(k)}\sin\nu-S_{\alpha}^{(k)}\cos\nu < 0
\]
on all intervals lying to the left of $\theta_{(m)}$. 
Using \eqref{C-S-not-0} again, and the conditions that $C_\alpha^{(m-1)}\ge0$, $S_\alpha^{(m-1)}\le0$, and $\theta_{(m)}<-\pi/2$, we also obtain
\[
I_{n,\alpha}(\theta_{(m)}) 
= C_\alpha^{(m-1)}\sin\theta_{(m)} - S_\alpha^{(m-1)}\cos\theta_{(m)}<0.
\]
On the other hand, for $\nu\in(\theta_{(m)},-\pi/2)$, the relevant coefficients satisfy
\[
C_{\alpha}^{(k)}<0, \quad S_{\alpha}^{(k)}<0, \quad \sin\nu<0, \quad \cos\nu<0.
\]
Thus, on each interval of constancy of the coefficients,
\[
I'_{n,\alpha}(\nu) = C_{\alpha}^{(k)}\cos\nu+S_{\alpha}^{(k)}\sin\nu >0 .
\]
Since $I_{n,\alpha}$ is continuous, it follows that $I_{n,\alpha}$ is strictly increasing on
$(\theta_{(m)},-\pi/2)$. Furthermore,
\[
I_{n,\alpha}(-\pi/2) = -C_{\alpha}^{(k_*)}>0.
\]
Consequently, $I_{n,\alpha}$ has exactly one zero in $(\theta_{(m)},-\pi/2)$, and no zero in
$(-\pi,\theta_{(m)}]$. Hence there is exactly one zero in $(-\pi,-\pi/2)$.

We next show that there are no zeros in $(-\pi/2,\pi/2)$. Let $k^+$ be the unique index such that
\[
\pi/2\in(\theta_{(k^+)},\theta_{(k^++1)}).
\]
Since $C_{\alpha}^{(k)}$ is strictly increasing on the central range $(-\pi/2,\pi/2)$, while $C_{\alpha}^{(k_*)}<0$, and since $C_{\alpha}^{(k^+)}>0$, there is exactly one sign change of
$C_{\alpha}^{(k)}$ in this range. The positivity of $C_{\alpha}^{(k^+)}$ follows from the fact that
$C_{\alpha}^{(k)}$ is decreasing on the right tail and ends at $C_{\alpha}^{(n)}>0$.
On $(-\pi/2,\pi/2)$, $J_{n,\alpha}(\nu) = I_{n,\alpha}(\nu)/\cos\nu$ is continuous and, on $(\theta_{(k)},\theta_{(k+1)})$,
\[
J'_{n,\alpha}(\nu) = \frac{C_{\alpha}^{(k)}}{\cos^2\nu}.
\]
Hence $J_{n,\alpha}$ first decreases and then increases. Its minimum is therefore attained either on an interval where $C_{\alpha}^{(k)}=0$, or at a sample point where $C_{\alpha}^{(k)}$ changes sign from negative to positive.
If $C_{\alpha}^{(k)}=0$ on such an interval, then $J_{n,\alpha}(\nu)=-S_{\alpha}^{(k)}>0$.

Otherwise, let $\theta_{(r)}\in(-\pi/2,\pi/2)$ be the sample point at which the sign changes, so that
\[
C_{\alpha}^{(r-1)}<0<C_{\alpha}^{(r)}.
\]
If $\theta_{(r)}<0$, then using the left representation gives
\[
J_{n,\alpha}(\theta_{(r)}) = C_{\alpha}^{(r-1)}\tan\theta_{(r)} - S_{\alpha}^{(r-1)} >0,
\]
because $C_{\alpha}^{(r-1)}<0$, $\tan\theta_{(r)}<0$, and $S_{\alpha}^{(r-1)}<0$. 
If $\theta_{(r)}>0$, then using the right representation gives
\[
J_{n,\alpha}(\theta_{(r)}) = C_{\alpha}^{(r)}\tan\theta_{(r)} - S_{\alpha}^{(r)} > 0.
\]
Thus the minimum of $J_{n,\alpha}$ on $(-\pi/2,\pi/2)$ is strictly positive, and it follows that
\[
I_{n,\alpha}(\nu)>0 \quad \text{for all } \nu\in(-\pi/2,\pi/2).
\]
Finally, consider $(\pi/2,\pi)$. On the right tail, $C_{\alpha}^{(k)}>0$ and $S_{\alpha}^{(k)}\le 0$,
because $C_{\alpha}^{(k)}$ is decreasing there and ends at $C_{\alpha}^{(n)}>0$, while
$S_{\alpha}^{(k)}<0$ for $k=1,\ldots,n-1$ and $S_{\alpha}^{(n)}=0$. Therefore, for $\nu\in(\pi/2,\pi)$,
\[
I'_{n,\alpha}(\nu) = C_{\alpha}^{(k)}\cos\nu+S_{\alpha}^{(k)}\sin\nu < 0.
\]
Since $\lim_{\nu\uparrow\pi}I_{n,\alpha}(\nu)=0$, we obtain
\[
I_{n,\alpha}(\nu)>0 \quad \text{for all } \nu\in(\pi/2,\pi).
\]
Thus, in Case 3, $I_{n,\alpha}$ has exactly one zero, this zero lies
in $(-\pi,-\pi/2)$, and $I_{n,\alpha}$ changes sign from negative to
positive at this zero.

Thus in all cases there exists exactly one zero $q_n(\alpha)\in(-\pi,0)$, and its position relative to $-\pi/2$ is determined by the sign of $C_\alpha^{(k_*)}$.

Finally, since $g_{n,\alpha}$ is differentiable on $(-\pi,\pi)$ with
$\partial g_{n,\alpha}(\nu,\theta)/\partial \nu=I_{n,\alpha}(\nu,\theta)$,
and since $I_{n,\alpha}$ changes sign from negative to positive at $q_n(\alpha)$,
it follows that $q_n(\alpha)$ is the unique minimizer of $g_{n,\alpha}(\nu,\theta)$.
\end{proof}


\begin{proof}[Proof of Proposition~\ref{consistency}]
By rotation invariance, assume $\mu=0$. We write $X$ for the $0$-centered representative, so that $X\in(-\pi,\pi)$ a.s. Put $q_\alpha=\mu_\alpha$ in this chart, and let $M_n:=\hat\mu_n$ be the sample circular mean represented in $(-\pi,\pi)$. 
Since $R>0$, the sample circular mean is strongly consistent, and hence
\[
M_n\to0 \quad\text{a.s.} 
\]
For a center $m$, let
\[
x^{(m)} = m + \operatorname{Arg}_{(-\pi,\pi)} \left(e^{i(x-m)}\right)
\]
denote the representative of $x$ in $(m-\pi,m+\pi)$. For $\nu\in(m-\pi,m+\pi)$, define
\[
\psi_m(\nu,x) = \alpha\sin(\nu-x^{(m)})\mathbf 1_{\{x^{(m)}>\nu\}}
+ (1-\alpha)\sin(\nu-x^{(m)})\mathbf 1_{\{x^{(m)}\le\nu\}} .
\]
With $P_n$ denoting the empirical measure, write
\[
I_{n,\alpha}^{(m)}(\nu) = P_n\psi_m(\nu,\cdot), \quad I_\alpha^{(0)}(\nu,F) = P\psi_0(\nu,\cdot).
\]
We first claim that, for every fixed $\nu\in(-\pi,\pi)$,
\[
I_{n,\alpha}^{(M_n)}(\nu)-I_{n,\alpha}^{(0)}(\nu) \to 0 \quad \text{a.s.}
\]
Indeed, choose $\eta>0$ such that $\nu\in(-\pi+\eta,\pi-\eta)$, and for $0<\delta<\eta$ set
\[
A_\delta=(-\pi,-\pi+\delta]\cup[\pi-\delta,\pi).
\]
If $|m|\le\delta$, then changing the branch center from $0$ to $m$ can change the representative only for observations in $A_\delta$. Since $|\psi_m|\le1$,
\[
\left|I_{n,\alpha}^{(M_n)}(\nu)-I_{n,\alpha}^{(0)}(\nu)\right|
\le 2P_n(A_\delta)
\]
eventually almost surely. Hence
\[
\limsup_{n\to\infty} \left|I_{n,\alpha}^{(M_n)}(\nu)-I_{n,\alpha}^{(0)}(\nu)\right|
\le 2P(A_\delta) \quad\text{a.s.}
\]
Letting $\delta\downarrow0$, using absolute continuity, proves the claim.
The strong law of large numbers gives
\[
I_{n,\alpha}^{(0)}(\nu)\to I_\alpha^{(0)}(\nu,F) \quad\text{a.s.}
\]
Thus
\[
I_{n,\alpha}^{(M_n)}(\nu)\to I_\alpha^{(0)}(\nu,F) \quad\text{a.s.}
\]
Now fix $\varepsilon>0$ such that $q_\alpha\pm\varepsilon\in(-\pi,\pi)$. 
By the assumed sign change of the population identification function,
\[
I_\alpha^{(0)}(q_\alpha-\varepsilon,F)<0 < I_\alpha^{(0)}(q_\alpha+\varepsilon,F).
\]
By the pointwise convergence just proved,
\[
I_{n,\alpha}^{(M_n)}(q_\alpha-\varepsilon)<0 < I_{n,\alpha}^{(M_n)}(q_\alpha+\varepsilon)
\]
eventually almost surely. Also, $q_\alpha\pm\varepsilon\in (M_n-\pi,M_n+\pi)$ eventually.

By Theorem~\ref{prop:empirical}, applied after centering at $M_n$, $I_{n,\alpha}^{(M_n)}$ has a unique zero $q_n(\alpha) \in (M_n-\pi,M_n)$, and this zero is the empirical circular $\alpha$-expectile.
Since $I_{n,\alpha}^{(M_n)}$ is continuous and changes sign between
$q_\alpha-\varepsilon$ and $q_\alpha+\varepsilon$, this zero lies between these two points. Hence, eventually almost surely,
\[
q_\alpha-\varepsilon < \hat\mu_{n,\alpha}^{(0)} < q_\alpha+\varepsilon,
\]
where $\hat\mu_{n,\alpha}^{(0)}$ is the $0$-centered representative of the empirical circular expectile. Since $\varepsilon>0$ was arbitrary,
\[
\hat\mu_{n,\alpha}^{(0)}\to q_\alpha \quad\text{a.s.}
\]
\end{proof}

\subsection*{Declaration of AI Use:} 
The authors acknowledge the use of OpenAI’s ChatGPT-5.5 Pro (chatgpt.com) to assist in the preparation of this manuscript. The AI tool was used to discuss statistical methodologies and evaluate the feasibility of alternative proof strategies. Furthermore, it was used for proof checking, correcting typographical errors in the derivations, and editing the manuscript text. Additionally, building upon a basic version of the computational implementation originally drafted by the author, the AI tool was used to implement input validation checks and optimize execution speed. The correctness of the AI-optimized code was verified by benchmarking its outputs against the original code and conducting independent tests. The author maintains sole accountability for the correctness and adequacy of the arguments and results, as well as for the completeness and accuracy of citations to relevant prior work.

\bibliographystyle{apalike}
\bibliography{CircularExpectiles-biblio}

\end{document}